\documentclass[conference]{IEEEtran}

\usepackage[T1]{fontenc}
\usepackage{tikz}
\usetikzlibrary{arrows}
\usetikzlibrary{calc}
\usetikzlibrary{snakes}
\usepackage{cite}
\usepackage{graphicx}
\usepackage{array}
\usepackage{mdwmath}
\usepackage{mdwtab}
\usepackage{eqparbox}
\usepackage{fixltx2e}
\usepackage{url}

\usepackage{amsmath,amsfonts,amssymb,amsthm}
\usepackage{xspace}
\usepackage{xcolor}
\usepackage{algorithmic}
\usepackage{bm}

\usepackage{booktabs}
\usepackage{nicefrac}
\usepackage{textcomp}
\usepackage{stmaryrd}
\usepackage{mathrsfs}
\usepackage{paralist}

\usepackage[font=footnotesize,labelsep=space,labelfont=bf]{caption}
\usepackage{url}

\usepackage[pdftex,bookmarks=true,pdfstartview=FitH,colorlinks,linkcolor=blue,filecolor=blue,citecolor=blue,urlcolor=blue]{hyperref}
\makeatletter

\newcommand{\Rmnum}[1]{\expandafter\@slowromancap\romannumeral #1@}
\makeatother

\newtheorem{thm}{Theorem}
\newtheorem{defn}{Definition}
\newtheorem{lem}{Lemma}
\newtheorem{claim}{Claim}
\newtheorem{cor}[thm]{Corollary}
\newtheorem{remark}{Remark}

\newcommand{\namedref}[2]{\hyperref[#2]{#1~\ref*{#2}}}

\newcommand{\EQ}{\text{EQ}}
\newcommand{\PS}{\text{PS}}
\newcommand{\AS}{\text{AS}}

\newcommand{\image}{\text{image}}

\newcommand{\Sectionref}[1]{\namedref{Section}{sec:#1}}

\newcommand{\Appendixref}[1]{\namedref{Appendix}{app:#1}}
\newcommand{\Theoremref}[1]{\namedref{Theorem}{thm:#1}}
\newcommand{\Corollaryref}[1]{\namedref{Corollary}{cor:#1}}

\newcommand{\Definitionref}[1]{\namedref{Definition}{defn:#1}}
\newcommand{\Lemmaref}[1]{\namedref{Lemma}{lem:#1}}
\newcommand{\Remarkref}[1]{\namedref{Remark}{remark:#1}}
\newcommand{\Claimref}[1]{\namedref{Claim}{claim:#1}}

\newcommand{\Figureref}[1]{\namedref{Figure}{fig:#1}}

\newcommand{\Pageref}[1]{\hyperref[#1]{page~\pageref*{#1}}}

\definecolor{darkred}{rgb}{0.5, 0, 0} 
\definecolor{darkgreen}{rgb}{0, 0.5, 0} 
\definecolor{darkblue}{rgb}{0,0,0.5} 
\hypersetup{
    colorlinks=true,
    linkcolor=darkred,
    citecolor=darkblue,
    urlcolor=darkblue   
}

\newcommand{\X}{\ensuremath{\mathcal{X}}\xspace}
\newcommand{\Y}{\ensuremath{\mathcal{Y}}\xspace}
\newcommand{\Z}{\ensuremath{\mathcal{Z}}\xspace}

\renewcommand{\paragraph}[1]{\smallskip\noindent{\bf #1}~}

\begin{document}

\title{Secure Computation of Randomized Functions}
\author{\IEEEauthorblockN{Deepesh Data}
\IEEEauthorblockA{School of Technology \& Computer Science\\
Tata Institute of Fundamental Research,
Mumbai, India\\
Email: deepeshd@tifr.res.in}} 
\maketitle

\pagestyle{plain}

\begin{abstract}
Two user secure computation of randomized functions is considered, where only one user computes the output. Both the users are semi-honest; and computation is such that no user learns any additional information about the other user's input and output other than what cannot be inferred from its own input and output. First we consider a scenario, where privacy conditions are against both the users. In perfect security setting Kilian gave a characterization of securely computable {\em randomized} functions in \cite{Kilian00}, and we provide rate-optimal protocols for such functions. We prove that the same characterization holds in asymptotic security setting as well and give a rate-optimal protocol. In another scenario, where privacy condition is only against the user who is not computing the function, we provide rate-optimal protocols. For perfect security in both the scenarios, our results are in terms of chromatic entropies of different graphs. In asymptotic security setting, we get single-letter expressions of rates in both the scenarios.
\end{abstract}

\section{Introduction}
Secure computation allows mutually distrustful users to collaborate and compute functions of their private data without revealing any additional information about their data to other users besides the function value. 
It has many real life applications such as electronic voting, privacy-preserving data mining, privacy-preserving machine learning, to name a few \cite{CramerDaNi15}. Essentially, we can cast any distributed computation task in a secure computation framework, where input data is distributed among many users, and they want to jointly compute some function of their data.
The study of two user secure computation was initiated in \cite{ShamirRiAd79}, 
\cite{Rabin81}, \cite{Yao82}, \cite{Yao86}, etc., where secure computation was shown to be feasible under cryptographic assumptions when the users are computationally bounded. In the information-theoretic setting, a combinatorial characterization of securely computable {\em deterministic} functions is given in \cite{Kushilevitz92,Beaver89a} for perfect security setting and in \cite{MajiPrRo09} for statistical security setting.
An alternative characterization (in perfect security setting) using common randomness generated by determinstic protocols is given in \cite{NarayanTyWa15}, which also provides lower bounds on the communication complexity.
Kilian~\cite{Kilian00}, among many other things, gave a characterization of securely computable {\em randomized} functions for the special case where only one user computes the output. The fundamental problem of characterizing securely computable {\em randomized} functions when both the users compute the outputs remains open.

Information theoretic tools have been used in deriving communication lower bounds in 3-user secure computation \cite{DataPrPr15}.
Another line of work, where communication among users is public and privacy of function value is against an eavesdropper, characterizes securely computable functions in terms of information-theoretic quantities \cite{TyagiNaGu11, Tyagi12}. A full characterization of joint distributions that can be securely sampled among two users is given in \cite{WangIs11}.

All these results are for {\em semi-honest} users. A semi-honest user follow the protocol honestly but tries to learn additional information about other user's input and output that cannot be inferred by its own input and output. 
In this paper we focus on two user secure computation of randomized functions against semi-honest users, where only one user computes the output. The secure computation problem is specified by a pair $(p_{XY},p_{Z|XY})$, where $p_{XY}$ is the input distribution from which Alice (user-1) and Bob (user-2) get their inputs $X$ and $Y$, respectively, and $p_{Z|XY}$ specifies the output distribution. Our results can be summarized as follows:

In \Sectionref{twoSidedPrivacy} we consider a scenario where privacy is against both the users. In perfect security setting, a characterization of securely computable {\em randomized} functions was given by Kilian~\cite{Kilian00}. We give rate-optimal protocols for such functions. The rates are in terms of {\em chromatic entropy} of a graph.
We prove that the same characterization holds in asymptotic security setting as well and give a rate-optimal protocol. Our achievability follows the Slepian-Wolf coding scheme \cite{SlepianWo73}. In order to use their scheme we modify our secure computation problem $(p_{XY},p_{Z|XY})$, such that the modified problem preserves the functionality, and when applied with their achievable scheme, does not leak any privacy.
In \Sectionref{oneSidedPrivacy} we consider a scenario where privacy is only against Alice. We give rate-optimal protocols for both the perfect security and the asymptotic security settings. In perfect security setting, the rates are in terms of chromatic entropy of a {\em different} graph, and in asymptotic security the rate is equal to {\em conditional graph entropy} of this graph.

\section{Preliminaries}\label{sec:prelims}
\paragraph{Protocol.} Alice and Bob get $X$ and $Y$ as inputs, respectively; then they engage in a {\em protocol}; and in the end Bob outputs according to $p_{Z|XY}$. No external eavesdropper/adversary is assumed. The link between the users is assumed to be noiseless and private. The users have access to private randomness. They communicate over multiple rounds, and in one round only one user sends message to the other. A protocol $\Pi$ is a set of triples $(\Pi_A,\Pi_B,\Pi_B^{\text{out}})$, where $\Pi_A$ and $\Pi_B$ are next message functions for Alice and Bob, respectively, and $\Pi_B^{\text{out}}$ is the output function for Bob. More precisely, at any round $t$ if it is Alice's turn to send a message to Bob, this new message is determined by the function $\Pi_A$ that takes Alice's input, her private randomness, and all the messages she has exchanged with Bob so far, and outputs a message. Similar is the case for Bob. In the end, Bob outputs according to the function $\Pi_B^{\text{out}}$ that takes Bob's input, his private randomness, and all the messages he has exchanged so far with Alice. We require that the messages in each round $t$ belong to a prefix-free code $\mathcal{C}_t$, which can be determined by all the messages exchanged so far. This is a natural requirement and allows us to lower-bound the expected length of the transcript by its entropy.

\paragraph{Notation.} For any $n\in\mathbb{N}$, we define $[n]:=\{1,2,\hdots,n\}$. We abbreviate {\em independent and identically distributed} by {\em i.i.d.}. 
For a random variable $U$, we denote by $U^n$ the $n$-length vector $(U_1,U_2,\hdots,U_n)$, where $U_i$'s are i.i.d. and distributed according to $U$. 
For $(U_1,U_2,\hdots,U_n)$ and for some $i\in[n]$, we define $U_{-i}$ to denote $(U_1,U_2,\hdots,U_{i-1},U_{i+1},\hdots,U_n)$.

\section{Privacy against Alice \& Bob, and Bob computes}\label{sec:twoSidedPrivacy}
\subsection{Perfect Security}\label{subsec:ps_twoSidedPrivacy}
We consider a two user secure computation problem specified by $(p_{XY},p_{Z|XY})$. 
Alice gets $X$ as her input; Bob gets $Y$ as his input; and Bob wants to 
securely compute an output $Z$, which should be distributed according to $p_{Z|XY}$; see \Figureref{ps_twoSidedPrivacy}. 
$X,Y,Z$ take values in finite alphabets $\X,\Y,\Z$, respectively. 
We assume that the p.m.f.s $p_X$, $p_Y$, and $p_Z$ have full support. 
The users compute via interactive communication over many rounds. Let $M$ denote the transcript by the end of their communication, which is concatenation of all the messages exchanged between them. Bob's output $Z$ depends on $X$ only via $(Y,M)$, i.e., $Z-(Y,M)-X$ is a Markov chain.
Here privacy is against a single user, which means that Alice should not learn anything from the protocol execution about $(Y,Z)$ other than what is revealed by her input, and similarly Bob should not learn anything from the protocol execution about $X$ other than what is revealed by his input $Y$ and output $Z$.
We say that a protocol for $(p_{XY},p_{Z|XY})$ is perfectly secure, if it computes the output
with zero error and it is perfectly private against any single user. A formal definition of 
a secure protocol is given below.

\begin{defn}\label{defn:ps_twoSidedPrivacy}
A protocol $\Pi$ for a secure computation problem $(p_{XY},p_{Z|XY})$ is perfectly secure, if 
the output $Z$ is distributed according to $p_{Z|XY}$, and it satisfies the following conditions:
\begin{align}
M-X-(Y,Z), \label{eq:ps_2_privacy-alice}\\
M-(Y,Z)-X. \label{eq:ps_2_privacy-bob}
\end{align}
\end{defn}
The Markov chain conditions \eqref{eq:ps_2_privacy-alice} and \eqref{eq:ps_2_privacy-bob} correspond to privacy against Alice and Bob, respectively.
It turns out that not every pair $(p_{XY},p_{Z|XY})$ can be computed securely. 
Kilian~\cite{Kilian00} gave a characterization of securely computable $(p_{XY},p_{Z|XY})$ with $p_{XY}$ having full support. Essentially the same characterization holds for general $(p_{XY},p_{Z|XY})$ as well, and we prove it in our language in \Theoremref{characterization}. It will be useful in understanding rest of the paper.
\begin{figure}
\centering
\begin{tikzpicture}
\draw (0,0) rectangle node {A} +(1,1); \draw [->,thick] (-1,0.5) -- (0,0.5); \node at (-0.5,0.7    5) {$X$};
\draw (4,0) rectangle node {B} +(1,1); \draw [<-,thick] (5,0.5) -- (6,0.5); \node at (5.5,0.75)     {$Y$};
\draw [->,thick] (4.5,0) -- (4.5,-0.75); \node [right] at (4.85,-0.4) {$Z\sim p_{Z|XY}$};
\draw [<->,thick] (1,0.5) -- (4,0.5); \node at (2.5,0.75) {$M$};
\end{tikzpicture}
\caption{A two user secure computation problem: Alice and Bob gets $X$ and $Y$ as inputs, respectively, and Bob wants to compute an output $Z$ which should be distributed according to $p_{Z|XY}$. Privacy is against the users themselves, which means that Alice does not learn any additional information about $Y,Z$ other than what is revealed by $X$, and Bob does not learn any additional information about $X$ other than what is revealed by $Y,Z$. There are no external adversaries/eavesdroppers. Both the users can talk to each other over multiple rounds over a noiseless and bidirectional link.}
\label{fig:ps_twoSidedPrivacy}
\end{figure}
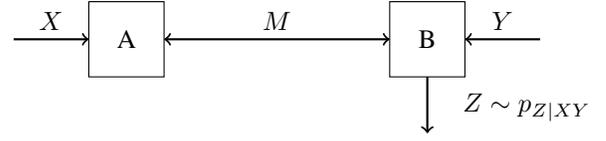
We need some definitions first.
\begin{defn}\label{defn:relation-tilde}
For any two distinct elements $x,x'$ of $\X$, we say that $x\sim x'$, if there exists $y\in \Y$
and $z\in \Z$ such that $p_{XY}(x,y)>0,p_{XY}(x',y)>0$ and $p_{Z|XY}(z|x,y)>0,p_{Z|XY}(z|x',y)>0$. 
\end{defn}
\begin{defn}\label{defn:equivalence-relation}
For any two distinct elements $x,x'$ of $\X$, we say that $x\equiv x'$, if there exists a sequence $x=x_1,x_2,\hdots,x_{l-1},x_l=x'$ for some integer $l$, where $x_i\sim x_{i+1}$ for every $i \in \{1,2,\hdots,l-1\}$.
\end{defn}
It can be verified easily that the above defined relation $\equiv$ is an equivalence relation.
We know that an equivalence relation partitions the whole space into equivalence classes.
Suppose the relation $\equiv$ partitions the space $\X$ as  
$\X=\mathcal{X}_1\biguplus \mathcal{X}_2\biguplus\hdots\biguplus\mathcal{X}_k$, 
where each $\mathcal{X}_i$ is an equivalence class.

We say that $C\times D$, where $C\subseteq \X, D\subseteq \Y$, is {\em column monochromatic}, if 
for every $y\in D$, if there exists distinct $x,x'\in C$ such that $p_{XY}(x,y),p_{XY}(x',y)>0$, then $p_{Z|XY}(z|x,y)=p_{Z|XY}(z|x',y)$, $\forall z\in\Z$.

\begin{thm}\label{thm:characterization}
Let $\X=\mathcal{X}_1\biguplus\mathcal{X}_2\biguplus\hdots\biguplus\mathcal{X}_k$ be the partition induced by the equivalence relation $\equiv$.
Then, $(p_{XY},p_{Z|XY})$ is computable with perfect security if and only if 
each $\mathcal{X}_i\times \Y$, $i\in\{1,2,\hdots,k\}$, is column monochromatic.
\end{thm}
\begin{proof}
We first prove the converse, i.e., if $(p_{XY},p_{Z|XY})$ is securely computable, then each $\mathcal{X}_i\times \Y$ is column monochromatic. For the other direction we give a secure protocol.

$\Longrightarrow$:
Fix a secure protocol $\Pi$. Let $p(x,y,m,z)$ be the resulting distribution.
Consider an equivalence class $\mathcal{X}_i$. 
First we prove that $p(m|x)=p(m|x')$ for every transcript $m$ and every $x,x'\in\mathcal{X}_i$.
Since $\mathcal{X}_i$ is an equivalence class of $\equiv$, we have $x\equiv x'$, 
and by the definition of $x\equiv x'$, there exists $x=x_1,x_2,\hdots,x_{l-1},x_l=x'$ 
such that $x_i\sim x_{i+1}$ for every $0\leq i \leq l-1$.
Consider $x_i,x_{i+1}$ for some $i$ in this sequence. Since $x_i\sim x_{i+1}$, there exists $y,z$ such that $p(x_i,y),p(x_{i+1},y),p(z|x_i,y),p(z|x_{i+1},y)>0$. Fix a transript $m$ and consider the following:
\begin{align*}
p(m|x_i) &= p(m|x_i,y,z) \\
&= p(m|x_{i+1},y,z) \\
&= p(m|x_{i+1}),
\end{align*}
where the first and third equalities follow from privacy against Alice \eqref{eq:ps_2_privacy-alice}, and the 
second equality follows from privacy against Bob \eqref{eq:ps_2_privacy-bob}.
Since the above argument holds for every $i\in\{1,2,\hdots,k-1\}$, we have $p(m|x)=p(m|x')$ for every $x,x'\in\mathcal{X}_i$.

Now take any $y\in\Y$. We prove that if there exists $x,x'\in\mathcal{X}_i$ such that 
$p_{XY}(x,y),p_{XY}(x',y)>0$, then $p_{Z|XY}(z|x,y)=p_{Z|XY}(z|x',y)$ for every $z\in\Z$.
Take any transcript $m$ such that $p(m|x),p(m|x')>0$, and take any $z\in\Z$. Consider $(x,y,z,m)$ and expand $p(m,z|x,y)$ as follows:
\begin{align}
p(m,z|x,y) &= p(m|x,y)p(z|x,y,m) \nonumber \\
&= p(m|x)p(z|m,y). \label{eq:ps_2_one-way}
\end{align}
In \eqref{eq:ps_2_one-way} we used privacy against Alice to write $p(m|x,y)=p(m|x)$, and 
the Markov chain $Z-(M,Y)-X$ to write $p(z|x,y,m)=p(z|y,m)$. Note that the Markov chain $Z-(M,Y)-X$ holds, because Bob outputs $Z$ by looking at his input $Y$ and the transcript $M$. 
We can expand $p(m,z|x,y)$ in another way:
\begin{align}
p(m,z|x,y) &= p(z|x,y)p(m|x,y,z) \nonumber \\
&= p(z|x,y)p(m|x). \label{eq:ps_2_another-way}
\end{align}
In \eqref{eq:ps_2_another-way} we used privacy against Alice to write $p(m|x,y,z)=p(m|x)$. 
Comparing \eqref{eq:ps_2_one-way} and \eqref{eq:ps_2_another-way} we get
\begin{align}
p(z|x,y)=p(z|m,y). \label{eq:ps_2_first-compare}
\end{align}
Running the same arguments with $(x',y,z,m)$ we get
\begin{align}
p(z|x',y)=p(z|m,y). \label{eq:ps_2_second-compare}
\end{align}
Comparing \eqref{eq:ps_2_first-compare} and \eqref{eq:ps_2_second-compare} gives 
$p(z|x,y)=p(z|x',y)$. Since the protocol is correct, i.e., $p(z|x,y)=p_{Z|XY}(z|x,y)$, we have our desired result that $p_{Z|XY}(z|x,y)=p_{Z|XY}(z|x',y)$. \\

$\Longleftarrow$: If $\mathcal{X}_i\times\Y$, $i\in\{1,2,\hdots,k\}$, is column monochromatic, 
then we give a simple protocol to securely compute $(p_{XY},p_{Z|XY})$ below. It is easy to check that the protocol is perfectly secure.

{\bf Protocol:} Alice sends the index $i$ such that $x\in\mathcal{X}_i$. Bob then samples $Z$ according to $p_i(z|y)$, where $p_i(z|y):=p_{Z|XY}(z|x,y)$ for some $x$ that satisfies $p_{XY}(x,y)>0$. 
\end{proof}

The above protocol may not be optimal in terms of communication complexity. 
Consider two equivalence classes $\mathcal{X}_i$ and $\mathcal{X}_j$. Let $\mathcal{Y}_i := \{y\in\Y : \exists x\in\mathcal{X}_i \text{ s.t. } p_{XY}(x,y)>0\}$; $\mathcal{Y}_j$ is defined similarly. Suppose $\mathcal{Y}_i\bigcap \mathcal{Y}_j=\phi$.
Then, Alice can send the same message whether $x\in\mathcal{X}_i$ or $x\in\mathcal{X}_j$, 
and still Bob will be able to correctly compute $p_{Z|XY}$.

For $i=1,2,\hdots,k$, define $\mathcal{M}_i$ to be the set of all possible transcripts $m$ when Alice's input is in $\mathcal{X}_i$.
\begin{claim}\label{claim:msgs_different-classes}
For some $i,j\in[k],i\neq j$, if there exists $x\in\mathcal{X}_i$ and $x'\in\mathcal{X}_j$ such that $p_{XY}(x,y)>0,p_{XY}(x',y)>0$ for some $y\in\Y$, then $\mathcal{M}_i\bigcap\mathcal{M}_j=\phi$.
\end{claim}
\begin{proof}
We prove the claim by contradiction.
Suppose there exists $i,j\in[k],i\neq j$, and $x\in\mathcal{X}_i$, $x'\in\mathcal{X}_j$ such 
that $p_{XY}(x,y)>0,p_{XY}(x',y)>0$ for some $y\in\Y$. Since $p_{XY}(x,y)>0$, there is a 
$z\in\Z$ with $p_{Z|XY}(z|x,y)>0$. Since $x$ and $x'$ are not in the same equivalence class, 
we have $p_{Z|XY}(z|x',y)=0$.

For the sake of contradiction, suppose $\mathcal{M}_i\bigcap\mathcal{M}_j\neq\phi$. 
Let $m\in\mathcal{M}_i\bigcap\mathcal{M}_j$. It follows from the proof of 
\Theoremref{characterization} that $p(m|x)>0$ and $p(m|x')>0$. Now consider $(x,y,z,m)$ and by 
expanding $p(m,z|x,y)$ in two different ways (similar to what we did in the proof of 
\Theoremref{characterization}) we get $p(z|x,y)=p(z|m,y)$. 
Applying the same arguments with $(x',y,z,m)$ gives $p(z|x',y)=p(z|m,y)$, 
which leads to a contradiction, because by assumption $p(z|x,y)>0$ but $p(z|x',y)=0$.
\end{proof}
From the proof of Theorem \ref{thm:characterization} we have that the transcript $M$ is conditionally independent of the input $X$ conditioned on the equivalence class to which $X$ belongs. This, together with \Claimref{msgs_different-classes} suggests replacing all the elements in $\mathcal{X}_i, i\in[k]$, by a single element $x_i$, which induces a new random variable, which we denote by $X_{\EQ}$. Note that $X_{\EQ}$ takes values in the set $\X_{\EQ}:=\{x_1,x_2,\hdots,x_k\}$, where $x_i$ is the representative of the equivalence class $\mathcal{X}_i$.
Thus we can define an equivalent problem ($p_{X_{\EQ},Y},p_{Z|X_{\EQ},Y}$) as follows:
\begin{itemize}
\item Define $p_{X_{\EQ},Y}(x_i,y):=\sum_{x\in\mathcal{X}_i}p_{XY}(x,y)$ for every $(x_i,y)\in(\X_{\EQ}\times\Y)$. 
\item For $(x_i,y)\in(\X_{\EQ}\times\Y)$, if there exists $x\in\X_i$ s.t. $p_{XY}(x,y)>0$, then define $p_{Z|X_{\EQ},Y}(z|x_i,y) :=p_{Z|XY}(z|x,y)$, for every $z\in\Z$. Otherwise define $p_{Z|X_{\EQ},Y}(z|x_i,y) :=0$, for every $z\in\Z$.
\end{itemize}
Note that the above definition is for $(p_{XY},p_{Z|XY})$ which are securely computable as per \Theoremref{characterization}.
\begin{remark}\label{remark:bob-always-learn}
{\em Observe that in any secure protocol for $(p_{XY},p_{Z|XY})$, if exists, Bob will {\em always} learn which equivalence class $\X_i$ Alice's input lies in. This is because $\forall (y,z)\in\Y\times\Z$ such that $p_{YZ}(y,z)>0$, there is a unique $x\in\X_{\EQ}$ s.t. $p_{X_{\EQ},Y}(x,y)>0$ and $p_{Z|X_{\EQ},Y}(z|x,y)>0$.}
\end{remark}
To give a communication optimal protocol, we define the following graph $G_{\EQ}=(V_{\EQ},E_{\EQ})$:
\begin{itemize}
\item $V_{\EQ}=\X_{\EQ}=\{x_1,x_2,\hdots,x_k\}$.
\item $E_{\EQ}=\{\{x_i,x_j\}: \text{ there exists } (y,z)\in\Y\times\Z \text{ such that } p_{X_{\EQ},Y}(x_i,y)>0, \quad p_{X_{\EQ},Y}(x_j,y) >0 \text{ and }p_{Z|X_{\EQ},Y}(z|x_i,y)\ \neq \ p_{Z|X_{\EQ},Y}(z|x_j,y)\}$
\end{itemize}
We say that a coloring $c:V\to \mathbb{N}$ of the vertices of a graph $G=(V,E)$ is {\em proper} if 
both the vertices of any edge $\{u,v\}$ have distinct colors, i.e., $c(u)\neq c(v)$.
\begin{claim}\label{claim:color-equiv-protocol}
Every proper coloring of the vertices of $G_{\EQ}$ corresponds to a secure protocol for 
$(p_{X_{\EQ}Y},p_{Z|X_{\EQ}Y})$, and every secure protocol for $(p_{X_{\EQ}Y},p_{Z|X_{\EQ}Y})$ corresponds to a collection of
proper colorings of the vertices of $G_{\EQ}$.
\end{claim}
\begin{proof}
$\Rightarrow:$ Let $c$ be a proper coloring of the vertices of $G_{\EQ}$ that Alice and Bob agree upon. Suppose Alice's input is $x\in\X_{\EQ}$. She sends the message $c(x)$ to Bob; and Bob can compute the function correctly. Note that there is no error in computation, because all the vertices that have the same color as $c(x)$ form an independent set in the graph $G_{\EQ}$, and thus are equivalent from the function computation point of view.
Privacy against Alice comes from the fact that for any fix coloring, color of a vertex is a deterministic function of Alice's input. For privacy against Bob, see \Remarkref{bob-always-learn}.

$\Leftarrow:$ It follows from privacy against Alice \eqref{eq:ps_2_privacy-alice}, that we can assume, without loss of generality, that any secure protocol requires only one message to be sent from Alice to Bob, and then Bob outputs. Fix a secure protocol $\Pi$. Note that $\Pi$ may be a randomized protocol, but once we fix the random coins of the users, the protocol becomes deterministic. And because $\Pi$ computes $(p_{X_{\EQ}Y},p_{Z|X_{\EQ}Y})$ with zero-error, all the possible random coins give a zero-error protocol. Fix random coins $\vec{r}=(\vec{r}_1,\vec{r}_2)$ of the users, and let $\Pi_{\vec{r}}$ denote the resulting deterministic protocol. Suppose Alice's input is $x\in\X_{\EQ}$, and she sends a message $m(x)$ to Bob. Since this protocol $\Pi_{\vec{r}}$ is with zero-error, the coloring $c_{\vec{r}}(x)=m(x), \forall x\in\X_{\EQ}$ will be a proper coloring of the vertices of $G_{\EQ}$. Run through all the possible random coins $\vec{r}$, this will produce a random coloring $(c_{\vec{r}})_{\vec{r}}$ of the vertices, where for any randomness $\vec{r}$, the corresponding coloring $c_{\vec{r}}$ is a proper coloring.
\end{proof}
The {\em chromatic entropy} $H_{\X}(G,X)$ of a probabilistic graph $G=(V,E)$, where $X$ is the distribution on the vertex set, was defined in \cite{AlonOr96} as follows. Below, $c(X)$ denotes the induced distribution on colors by the coloring of the graph.
\[H_{\X}(G,X):=\min\{H(c(X)): c(X) \text{ is a proper coloring of G}\}.\]
\begin{defn}\label{defn:rate-defn1_ps2}
Suppose $(p_{XY},p_{Z|XY})$ is computable with perfect security. Then, for any perfectly secure protocol $\Pi$, we define its rate $R^{\PS_1}$ to be the expected length of the transcript generated by the protocol, i.e., $R^{\PS_1}:=\mathbb{E}[|M|]$.
\end{defn}
\begin{thm}\label{thm:ps_2_rate-bounds}
Any secure protocol for $(p_{XY},p_{Z|XY})$, if exists, satisfies the following bounds on its 
rate: \[H_{\X}(G_{\EQ},X_{\EQ}) \leq R^{\PS_1} < H_{\X}(G_{\EQ},X_{\EQ}) + 1.\]
\end{thm}
\begin{proof}
Upper bound: By \Claimref{color-equiv-protocol}, optimal coloring of $G_{\EQ}$ corresponds to a secure protocol for $(p_{XY},p_{Z|XY})$. So the expected length of the transcript (which is equal to the rate of the protocol) is upper-bounded by $H_{\X}(G_{\EQ},X_{\EQ})+1$.

Lower bound: By \Claimref{color-equiv-protocol}, every secure (randomized) protocol also corresponds to proper (random) colorings $(c_{\vec{r}})_{\vec{r}}$ of the graph. Since expected length $L(c_{\vec{r}})$ of any coloring $c_{\vec{r}}$ is lower-bounded by $H_{\X}(G_{\EQ},X_{\EQ})$, it follows that $\mathbb{E}_{\vec{r}}[L(c_{\vec{r}})]\geq H_{\X}(G_{\EQ},X_{\EQ})$. So the rate of communication is also lower-bounded by $H_{\X}(G_{\EQ},X_{\EQ})$.
\end{proof}

{\bf Multiple Instances: } 
Alice and Bob have blocks of inputs $X^n$ and $Y^n$, respectively, where $(X_i,Y_i)\sim p_{XY}$, i.i.d., and Bob wants to compute $Z^n$, where $Z^n$ is distributed according to $p_{Z^n|X^nY^n}(z^n|x^n,y^n)=\Pi_{i=1}^n p_{Z|XY}(z_i|x_i,y_i)$. It is easy to see that $(p_{X^nY^n},p_{Z^n|X^nY^n})$ is securely computable if and only if $(p_{XY},p_{Z|XY})$ satisfies the condition of \Theoremref{characterization}. We can define the graph $G_{\EQ}^n=(V_{\EQ}^n,E_{\EQ}^n)$ by extending the definition of $G_{\EQ}$ in a straightforward way: $V_{\EQ}^n=\X_{\EQ}^n$ -- the $n$-times cartesian product of $\X_{\EQ}$, and $E_{\EQ}^n=\{(x^n,x'^n)\in\X_{\EQ}^n\times \X_{\EQ}^n: \exists (y^n,z^n)\in\Y^n\times\Z^n\text{ s.t. } p_{X_{\EQ}^nY^n}(x^n,y^n)>0, p_{X_{\EQ}^nY^n}(x'^n,y^n)>0 \text{ and } p_{Z^n|X_{\EQ}^nY^n}(z^n|x^n,y^n)\neq p_{Z^n|X_{\EQ}^nY^n}(z^n|x'^n,y^n)\}$. It turns out that we can write an equivalent but simpler definition of $E_{\EQ}^n$ as follows: we put an edge between $(x_{1},x_{2},\hdots,x_{n})\in\X_{\EQ}^n$ and $(x'_{1},x'_{2},\hdots,x'_{n})\in\X_{\EQ}^n$ if there exists a $y^n\in \Y^n$ such that the following two conditions hold:
\begin{enumerate}
\item $\forall k\in[n], p_{X_{\EQ},Y}(x_{k},y_k)>0,p_{X_{\EQ},Y}(x'_{k},y_k)>0$.
\item $\exists k\in[n]$ and $\exists z\in\Z$ s.t. $p_{Z|X_{\EQ},Y}(z|x_{k},y_k)\neq p_{Z|X_{\EQ},Y}(z|x'_{k},y_k)$.
\end{enumerate}
Equivalence of the two definitions: Clearly, the first definition implies the second one. To see the other direction, let $x^n,x'^n,y^n,z$ be such that the above two conditions hold. Now consider an index $k\in[n]$ such that $p(z|x_k,y_k)\neq p(z|x_k,y_k)$, and, without loss of generality, assume that $p(z|x_k,y_k)> p(z|x_k,y_k)$. We need to find $z^n$ such that $p(z^n|x^n,y^n)\neq p(z^n|x'^n,y^n)$. Set $z_k=z$, and consider any $i\in[n]\setminus\{k\}$. If $p(z|x_i,y_i)=p(z_i|x'_i,y_i),\forall z\in\Z$, then set $z_i$ to be equal to any $z$ for which $p(z|x_i,y_i)>0$. Otherwise, pick a $z$ such that $p(z|x_i,y_i)>p(z_i|x'_i,y_i)$ -- such a $z$ exists because the functions $p(.|x_i,y_i)$ and $p(.|x'_i,y_i)$ are unequal -- and set $z_i=z$. It follows that with this $z^n$ we have $p(z^n|x^n,y^n)\neq p(z^n|x'^n,y^n)$.

If $(p_{X^nY^n},p_{Z^n|X^nY^n})$ is computable with perfect security, then for any secure protocol $\Pi_n$, we can define -- analogous to the \Definitionref{rate-defn1_ps2} -- the rate $R_n^{\text{PS}}=\frac{1}{n}\mathbb{E}[|M|]$, where $M$ is the transcript generated by the protocol $\Pi_n$. 
\begin{cor}\label{cor:ps_2_rate-bounds-n}
Any secure protocol $\Pi_n$, if exists, for $(p_{X^nY^n},p_{Z^n|X^nY^n})$ satisfies the following bounds on its 
rate: \[\frac{1}{n} H_{\X}(G_{\EQ}^n,X_{\EQ}^n) \leq R_n^{\text{PS}} \le \frac{1}{n} (H_{\X}(G_{\EQ}^n,X_{\EQ}^n) + 1).\]
\end{cor}
Observe that $G_{\EQ}^n$ in neither the {\sc and}-product nor the {\sc or}-product of the graph $G_{\EQ}$.
And as far as we know, no single letter expression of $\frac{1}{n}H_{\X}(G_{\EQ}^n,X_{\EQ}^n)$ (not even of $\lim_{n\to\infty}\frac{1}{n}H_{\X}(G_{\EQ}^n,X_{\EQ}^n)$) is known.  See \cite{AlonOr96} for the definitions of {\sc and}-product and {\sc or}-product of graphs.

\subsection{Asymptotic Security}\label{subsec:as_twoSidedPrivacy}
We now consider asymptotically secure computation. Alice and Bob have inputs $X^n$ and $Y^n$, respectively, where $(X_i,Y_i)$ are i.i.d. and distributed according to $p_{XY}$. Bob computes $\hat{Z}^n\sim p_{\hat{Z}^n|X^nY^n}$ such that $\mathbb{E}[||p_{\hat{Z}^n|X^nY^n}-p_{Z^n|X^nY^n}||_1]\to0$ as $n\to\infty$. We also allow vanishing information leakage, i.e., the average amount of information any user obtains from the protocol about other user's data must go to zero as block-length goes to infinity. The rate $R^{\AS_1}$ of a protocol is defined as $R^{\AS_1}:=\mathbb{E}[|M|]$, where $M$ is the transcript generated by the protocol. A formal definition is given below.
\begin{defn}\label{defn:as_twoSidedPrivacy}
For a secure computation problem $(p_{XY},p_{Z|XY})$, the rate $R^{\AS_1}$ is achievable 
if there exists a sequence of protocols $\Pi_n$ with rate $R^{\AS_1}$ such that for 
every $\epsilon>0$, there is a large enough $n$ such that
\begin{align}
\mathbb{E}_{X^nY^n}||p_{\hat{Z}^n|X^nY^n}-p_{Z^n|X^nY^n}||_1 &\leq \epsilon, \label{eq:as_2_correctness} \\
I(M;Y^n,\hat{Z}^n|X^n) &\leq n\epsilon, \label{eq:as_2_privacy-alice} \\
I(M;X^n|Y^n,\hat{Z}^n) &\leq n\epsilon. \label{eq:as_2_privacy-bob}
\end{align}
\end{defn}
We first show that the set of $(p_{XY},p_{Z|XY})$ pairs that are computable with asymptotic 
security is the same as the set of $(p_{XY},p_{Z|XY})$ pairs that are computable with 
perfect security.
\begin{thm}\label{thm:as_characterization}
Let $\X=\mathcal{X}_1\biguplus\mathcal{X}_2\biguplus\hdots\biguplus\mathcal{X}_k$ be the partition induced by the equivalence relation $\equiv$.
Then, $(p_{XY},p_{Z|XY})$ is computable with asymptotic security if and only if 
each $\mathcal{X}_i\times \Y$, $i\in\{1,2,\hdots,k\}$, is column monochromatic.
\end{thm}
\begin{proof}
Note that in order to prove the theorem it suffices to prove that every $(p_{XY},p_{Z|XY})$ 
computable with asymptotic security is also computable with perfect security.
Assume that $(p_{XY},p_{Z|XY})$ is computable with asymptotic security, which means that for every $\epsilon>0$, there is a large enough $n$ and a protocol $\Pi_n$ that satisfies \eqref{eq:as_2_correctness}-\eqref{eq:as_2_privacy-bob}. As we show in \Lemmaref{as-implies-ss} in \Appendixref{app-twoSidedPrivacy}, this implies existence of another protocol $\Pi_{\epsilon'}$ (with possibly different transcript $M'$) that satisfies the following conditions for some $\epsilon'>0$, where $\epsilon'\to0$ as $\epsilon\to0$.
\begin{align}
\mathbb{E}_{XY}||p_{\hat{Z}|XY}-p_{Z|XY}||_1 &\leq \epsilon', \label{eq:ss_2_correctness} \\
I(M';Y,\hat{Z}|X) &\leq \epsilon', \label{eq:ss_2_privacy-alice} \\
I(M';X|Y,\hat{Z}) &\leq \epsilon'. \label{eq:ss_2_privacy-bob}
\end{align}
Now, for a given $(p_{XY},p_{Z|XY})$, consider the following set.
\begin{align}
\mathcal{S}_{\epsilon} := \{p_{M\hat{Z}|XY} \ &: \ I(M;Y,\hat{Z}|X)\leq \epsilon, \notag \\
&\hspace{0.6cm} I(M;X|Y,\hat{Z}) \leq \epsilon, \notag \\
&\hspace{0.9cm} \mathbb{E}_{XY}||p_{\hat{Z}|XY}-p_{Z|XY}||_1 \leq \epsilon \}\label{eq:set-for-ss}.
\end{align}
It follows that the above set $\mathcal{S}_{\epsilon}$ is non-empty for every $\epsilon>0$,
and if $\epsilon_1>\epsilon_2>\epsilon_3>\hdots$ where $\epsilon_i$'s form a monotonically decreasing sequence with $\lim_{i\to\infty}\epsilon_i=0$, then it can be verified easily that $\mathcal{S}_{\epsilon_1}\supseteq\mathcal{S}_{\epsilon_2}\supseteq\mathcal{S}_{\epsilon_3}\supseteq\hdots$
We prove in \Lemmaref{compactness} in \Appendixref{app-twoSidedPrivacy} (using continuity of mutual information and continuity of $L_1$-norm) that $\forall \epsilon>0$, the set $\mathcal{S}_{\epsilon}$ is compact. 
By Cantor's intersection theorem, which states that a decreasing nested sequence of non-empty 
compact sets has non-empty intersection, we have that $\bigcap_{i}\mathcal{S}_{\epsilon_i}\neq \phi$.
We show, below, that the set $\mathcal{S}_0$, which corresponds to the perfect security, is equal to $\lim_{k\to\infty}\bigcap_{i=1}^k\mathcal{S}_{\epsilon_i}$. This proves the theorem.

\[\mathcal{S}_0 = \lim_{k\to\infty}\bigcap_{i=1}^k\mathcal{S}_{\epsilon_i}.\]
$\mathcal{S}_0\subseteq \lim_{k\to\infty}\bigcap_{i=1}^k\mathcal{S}_{\epsilon_i}$: This trivially holds.

$\lim_{k\to\infty}\bigcap_{i=1}^k\mathcal{S}_{\epsilon_i} \subseteq \mathcal{S}_0$: Since $\bigcap_{i}\mathcal{S}_{\epsilon_i}\neq \phi$, there exists $q_{M\hat{Z}|XY}\in\bigcap_{i}\mathcal{S}_{\epsilon_i}$. 
Since $q_{M\hat{Z}|XY}$ satisfies \eqref{eq:set-for-ss} for every $\epsilon_k$ in the sequence, we have that $I(M;Y,\hat{Z}|X)|_{q_{M\hat{Z}|XY}} \leq \epsilon_k,\quad \forall k\in\mathbb{N}$. Now, the fact that mutual information is always non-negative and that $\lim_{k\to\infty}\epsilon_k=0$, we have $I(M;Y,\hat{Z}|X)|_{q_{M\hat{Z}|XY}}=0$. Similarly we can prove $I(M;X|Y,\hat{Z})|_{q_{M\hat{Z}|XY}} = 0$ and $\mathbb{E}_{XY}||p_{\hat{Z}|XY}-p_{Z|XY}||_1|_{q_{M\hat{Z}|XY}}=0$. The latter conclusion implies that $p_{XY\hat{Z}}=p_{XYZ}$, i.e., Bob's output is with correct distribution. So we can replace $\hat{Z}$ by $Z$ and get $I(M;Y,Z|X)=0$ and $I(M;X|Y,Z)= 0$. This implies that $q_{MZ|XY}$ belongs to the set $S_0$, which concludes the proof.
\end{proof}
\begin{thm}\label{thm:as_2_rate-region}
\[R^{\AS_1}=H(X_{\EQ}|Y).\]
\end{thm}
\begin{proof}
{\allowdisplaybreaks
{\bf Lower bound $R^{\AS_1}\geq H(X_{\EQ}|Y)$:} We prove lower bound for protocols with multiple rounds of interaction, where messages sent by any user may belong to a variable length prefix-free code. This implies that $\mathbb{E}[|M|] \geq H(M)$.
Since $R^{\AS_1}=\frac{1}{n}\mathbb{E}[|M|]$, it is enough to show that $H(M)\geq H(X_{\EQ}|Y)$.
\begin{align}
&H(M) \geq H(M|Y^n) \geq I(M;X^n,\hat{Z}^n|Y^n) \notag \\
&= \sum_{i=1}^n I(M;X^n,\hat{Z}_i|\hat{Z}^{i-1},Y^n) \notag \\
&\geq \sum_{i=1}^n I(M;X_i,\hat{Z}_i|\hat{Z}^{i-1},Y^n) \notag \\
&= \sum_{i=1}^n I(M,Y_{-i},\hat{Z}^{i-1};X_i,\hat{Z}_i|Y_i) \notag\\
&\hspace{3cm} - \sum_{i=1}^n I(Y_{-i},\hat{Z}^{i-1};X_i,\hat{Z}_i|Y_i) \label{eq:as2-interim-5} \\
&\geq \sum_{i=1}^n I(M_i;X_i,\hat{Z}_i|Y_i) - n\epsilon_3,\ \text{ where }M_i=(M,Y_{-i})\notag \\
&= n\cdot \sum_{i=1}^n \frac{1}{n}\cdot I(M_i;X_i,\hat{Z}_i|Y_i,T=i) - n\epsilon_3 \label{eq:as2-interim-3} \\
&= n\cdot I(M_T;X_T,\hat{Z}_T|Y_T,T) - n\epsilon_3 \notag \\
&= n\cdot I(M_T,T;X_T,\hat{Z}_T|Y_T) - n\epsilon_3 \label{eq:as2-interim-2}
\end{align}
Note that the second summation in \eqref{eq:as2-interim-5} becomes zero if we replace $\hat{Z}^n$ by $Z^n$. Now, since $\hat{Z}^n$ is close to $Z^n$, using continuity of mutual information this summation can be upper-bounded by $n\epsilon_3$, where $\epsilon_3\to0$ as $\epsilon\to0$. Proof is given in \Lemmaref{small} in \Appendixref{app-twoSidedPrivacy}.
The random variable $T$ in \eqref{eq:as2-interim-3} is independent of $(M,X^n,Y^n,Z^n)$ and is uniformly distributed in $\{1,2,\hdots,n\}$. Let $M'=(M_T,T)$.
Similar to the above single-letterization by which we obtain \eqref{eq:as2-interim-2}, we can prove the following also (proof is along the lines of the proof of \Lemmaref{as-implies-ss} in \Appendixref{app-twoSidedPrivacy}):
\begin{align*}
I(M;Y^n,\hat{Z}^n|X^n)\leq \epsilon &\implies I(M';Y_T,\hat{Z}_T|X_T)\leq \epsilon' \\
I(M;X^n|Y^n,\hat{Z}^n)\leq \epsilon &\implies I(M';X_T|Y_T,\hat{Z}_T)\leq \epsilon' \\
\hat{Z}^n-(M,Y^n)-X^n &\implies \hat{Z}_T-(M',Y_T)-X_T
\end{align*}
\begin{align*}
\mathbb{E}_{X^nY^n}[||p_{\hat{Z}^n|X^nY^n}-p_{Z^n|X^nY^n}&||_1] \leq \epsilon \implies\\
\mathbb{E}_{X_TY_T}||&p_{\hat{Z}_T|X_TY_T}-p_{Z_T|X_TY_T}||_1 \leq \epsilon'
\end{align*}
where $\epsilon'\to0$ as $\epsilon\to0$. Now continuing from \eqref{eq:as2-interim-2}:
\begin{align}
&\geq n\cdot \displaystyle \min_{\substack{p_{M'\hat{Z}_T|X_TY_T}: \\ I(M';Y_T,\hat{Z}_T|X_T)\leq \epsilon' \\ I(M';X_T|Y_T,\hat{Z}_T)\leq \epsilon' \\ \hat{Z}_T-(M',Y_T)-X_T \\ \mathbb{E}_{X_TY_T}||p_{\hat{Z}_T|X_TY_T}\\\hspace{1cm}-p_{Z_T|X_TY_T}||_1 \leq \epsilon'}} I(M';X_T,\hat{Z}_T|Y_T) - \epsilon_3 \notag \\
&\geq n\cdot \displaystyle \min_{\substack{p_{M'\hat{Z}|XY}: \\ I(M';Y,\hat{Z}|X)\leq \epsilon' \\ I(M';X|Y,\hat{Z})\leq \epsilon' \\ \hat{Z}-(M',Y)-X \\ \mathbb{E}_{XY}||p_{\hat{Z}|XY}-p_{Z|XY}||_1 \leq \epsilon'}} I(M';X,\hat{Z}|Y) - \epsilon_3 \label{eq:as2-interim} \\
&\geq n\cdot \displaystyle \min_{\substack{p_{U|X}: \\ U-X-(Y,Z) \\ U-(Y,Z)-X \\ Z-(U,Y)-X}} I(U;X,Z|Y) - \epsilon''- \epsilon_3 \label{eq:as2-interim1} \\
&\geq n\cdot \displaystyle \min_{\substack{p_{U|X}: \\ U-X-(Y,Z) \\ U-(Y,Z)-X \\ Z-(U,Y)-X}} I(U;X|Y) - \epsilon''- \epsilon_3 \label{eq:as2-interim1-5} \\
&\geq n\cdot \displaystyle \min_{\substack{p_{U|X_{\EQ}}: \\ U-X_{\EQ}-(Y,Z) \\ Z-(U,Y)-X_{\EQ}}} I(U;X_{\EQ}|Y) - \epsilon''- \epsilon_3 \label{eq:as2-interim2} \\
&= n\cdot (H(X_{\EQ}|Y)-\epsilon''-\epsilon_3) \label{eq:as2-interim3}
\end{align}
}
where $\epsilon''+\epsilon_3\to0$ as $\epsilon\to0$.

{\bf Explanation for \eqref{eq:as2-interim1}:} This follows from continuity of mutual information, continuity of $L_1$-norm, and the fact that $(p_{XY},p_{Z|XY})$ satisfies the characterization of \Theoremref{as_characterization}. Detailed proof is given in \Appendixref{app-twoSidedPrivacy}.

{\bf Explaination for \eqref{eq:as2-interim1-5}:} Using the Markov chain $U-X-(Y,Z)$ condition we can further simplify the objective function $I(U;X,Z|Y)$ as follows: $I(U;Y,Z|X)=I(U;Y|X)+I(U;Z|X,Y)$, where the second term in the right hand side is zero because of the Markov chain $U-X-(Y,Z)$. Note that $U-X-(Y,Z)\implies U-(X,Y)-Z$.

{\bf Explanation for \eqref{eq:as2-interim2}:} Note that for a fix $(p_{XY},p_{Z|XY})$, $X_{\EQ}$ is a function of $X$. From the proof of Theorem \ref{thm:characterization} we have that the transcript $M$ is conditionally independent of the input $X$ conditioned on the equivalence class to which $X$ belongs, which implies that the Markov Chain $U-X_{\EQ}-X$ holds. Now \eqref{eq:as2-interim2} follows from the following:
\begin{itemize}
\item $U-X-(Y,Z)$ together with $U-X_{\EQ}-X$ implies $U-X_{\EQ}-(Y,Z)$:
\begin{align*}
0&=I(U;Y,Z|X)=I(U;Y,Z,X)-I(U;X)\\
&=I(U;Y,Z,X,X_{\EQ})-I(U;X,X_{\EQ})\\
&=I(U;Y,Z,X|X_{\EQ}) + I(U;X_{\EQ})\\
&\hspace{2cm} -I(U;X_{\EQ})-\underbrace{I(U;X|X_{\EQ})}_{=\ 0}\\
&\geq I(U;Y,Z|X_{\EQ})
\end{align*}
\item Since $X_{\EQ}$ is a function of $X$, the Markov chains $U-(Y,Z)-X$ and $Z-(U,Y)-X$ imply $U-(Y,Z)-X_{\EQ}$ and $Z-(U,Y)-X_{\EQ}$, respectively.
\item Similarly, the objective function $I(U;X,Z|Y)$ can be lower-bounded by $I(U;X_{\EQ},Z|Y)$, which is equal to $I(U;X_{\EQ}|Y)$ because of the Markov chain $U-X_{\EQ}-(Y,Z)$ condition.
\end{itemize}
Notice that we omit the Markov chain $U-(Y,Z)-X_{\EQ}$ from the right hand side of \eqref{eq:as2-interim2}. This Markov chain is redundant because Bob always learns $X_{\EQ}$; see \Remarkref{bob-always-learn}.

{\bf Explanation for \eqref{eq:as2-interim3}:} It follows from the Markov chain $Z-(U,Y)-X_{\EQ}$ that $H(X_{\EQ}|U,Y)=H(X_{\EQ}|U,Y,Z)$, where $H(X_{\EQ}|U,Y,Z)=0$, because from $(Y,Z)$ Bob always learns $X_{\EQ}$ (see \Remarkref{bob-always-learn}). So we have $H(X_{\EQ}|U,Y)=0$, which implies that $I(U;X_{\EQ}|Y)=H(X_{\EQ}|Y)$.

\noindent{\bf Upper bound $R^{\AS_1}\leq H(X_{\EQ}|Y)$:}
Our achievability is along the lines of the coding scheme given by Slepian and Wolf \cite{SlepianWo73}, in which Alice sends only one message to Bob, and then Bob computes the function. It requires a rate of $H(X_{\EQ}|Y)$. Details follow:
Note that if the function computation is correct, then Bob always learns $X_{\EQ}^n$; see \Remarkref{bob-always-learn}. Since we allow small probability of error in function computation, it follows that Bob learns $X_{\EQ}^n$ with very high probability. So we design a scheme in which Bob, having side information $Y^n$, learns $X_{\EQ}^n$ with very high probability (and therefore computes the output $Z^n$ by sampling from the distribution $p_{Z^n|X_{\EQ}^nY^n}$, whenever he recovers $X_{\EQ}^n$ correctly). This is exactly the Slepian-Wolf problem \cite{SlepianWo73}, where a sender wants to send his input sequence to a receiver who has a side information. Notice that the coding scheme in \cite{SlepianWo73} trivially satisfies privacy against Alice, because it is designed for $p_{X_{\EQ}^nY^n}$ and is independent of $p_{Z^n|X_{\EQ}^nY^n}$.
\end{proof}
\section{Privacy against Alice, and Bob computes}\label{sec:oneSidedPrivacy}
In this section the setup is exactly the same as the setup of \Sectionref{twoSidedPrivacy}, except for that we do not have privacy against Bob here.
We can define secure protocols with perfect security and asymptotic security, and the corresponding rates, $R^{\PS_2}, R_n^{PS_2}$ and $R^{\AS_2}$, analogous to as we defined them in \Sectionref{twoSidedPrivacy}.

It turns out that any $(p_{XY},p_{Z|XY})$ can be computed securely: Alice sends $X$ to Bob, and Bob outputs $Z$ according to $p_{Z|XY}$. In order to state our results, we define the following graph $G=(V,E)$: The vertex set $V=\X$, and the edge set $E=\{(x,x'): \exists (y,z)\in\Y\times\Z\text{ s.t. }p_{XY}(x,y)>0,p_{XY}(x',y)>0 \text{ and } p_{Z|XY}(z|x,y)\neq p_{Z|XY}(z|x',y)\}$. For $(p_{X^nY^n},p_{Z^n|X^nY^n})$ we can define the graph $G^n$ in a similar fashion as we did in \Sectionref{twoSidedPrivacy}. 

An {\em independent set} of a graph $G=(V,E)$ is a collection $U\subseteq V$ of vertices such that no two vertices of $U$ are connected by an edge in $G$.
Let $W$ denote the random variable corresponding to the independent sets in $G$, and let $\varGamma(G)$ be the set of all independent sets in $G$. The conditional graph entropy of $G$ is defined as follows:
\begin{defn}[Conditional Graph Entropy, \cite{OrlitskyRoche}]\label{defn:conditional-graph-entropy}
For a given $(p_{XY},p_{Z|XY})$, the conditional graph entropy of $G$ (as defined above), is defined as
\[H_{G}(X|Y):= \min_{\substack{p_{W|X}: \\ W-X-Y \\ X\in W\in \varGamma(G)}}I(W;X|Y),\]
where minimum is taken over all the conditional distributions $p_{W|X}$ such that $p(w|x)>0$ only if $x\in w$.
\end{defn}
\begin{thm}\label{thm:rate-bounds_privacy-alice}
For a given $(p_{XY},p_{Z|XY})$, the following bounds hold:
\begin{align}
H_{\X}(G,X) \leq &R^{\PS_2} < H_{\X}(G,X) + 1, \label{eq:ps_1_rate-bounds} \\
\frac{1}{n} H_{\X}(G^n,X^n) \leq &R_n^{\PS_2} \le \frac{1}{n} (H_{\X}(G^n,X^n) + 1), \label{eq:ps_1_rate-bounds-n} \\
& R^{\AS_2} = H_G(X|Y). \label{eq:as_1_rate-region}
\end{align}
\end{thm}
\begin{proof}
Bounds on $R^{\PS_2}$ in \eqref{eq:ps_1_rate-bounds} and $R_n^{\PS_2}$ in \eqref{eq:ps_1_rate-bounds-n} can be proved along the lines of the proofs of \Theoremref{ps_2_rate-bounds} and \Corollaryref{ps_2_rate-bounds-n}, respectively.
Now we prove \eqref{eq:as_1_rate-region}.

{\bf Lower bound $R^{\AS_2}\geq H_G(X|Y)$:} 
This can be proved along the lines of the lower bound proof of \Theoremref{as_2_rate-region}, except that we do not have privacy against Bob here. Note that our lower bound is for protocols with multiple rounds of interaction. Following the proof in \Theoremref{as_2_rate-region} until \eqref{eq:as2-interim1-5} gives:
\begin{align}
H(M)&\geq n\cdot \displaystyle \min_{\substack{p_{U|X}: \\ U-X-(Y,Z) \\ Z-(U,Y)-X}} I(U;X|Y) - \epsilon''- \epsilon_3 \notag\\
&= n\cdot \displaystyle \min_{\substack{p_{W|X}: \\ W-X-Y \\ X\in W\in \varGamma(G)}} I(W;X|Y) - \epsilon''- \epsilon_3 \label{eq:as1-interim3} \\
&= n\cdot (H_G(X|Y)-\epsilon''-\epsilon_3) \notag
\end{align}
By taking $\epsilon\to0$, we get $R^{\AS_2}\geq H_G(X|Y)$. \eqref{eq:as1-interim3} is proved in \Lemmaref{cond-graph-entropy-equi} in \Appendixref{app-oneSidedPrivacy}

{\bf Upper bound $R^{\AS_2}\leq H_G(X|Y)$:} In our achievabe scheme Alice sends only one message to Bob, and then Bob computes the function. Achievability is along the lines of the coding scheme given by Orlitsky and Roche \cite{OrlitskyRoche}. They showed, by a random coding argument, existence of a sequence of codebooks with rate approaching $H_G(X|Y)$, in which probability of error in function computation goes to zero as $n\to\infty$. Fix such a sequence of codebooks. It is shown in \cite{OrlitskyRoche}, that with high probability, after decoding the message Bob obtains a sequence $W^n$ of independent sets in $G$ that is robustly jointly typical with Alice's input sequence $X^n$, which implies that $X_i\in W_i, \forall i\in[n]$. By definition, if $w$ is an independent set in $G$, then $\forall (y,z)\in\Y\times\Z$, $p_{Z|XY}(z|x,y)$ is same for all $x$'s with $p(x,y)>0$, which means that $z$ can be computed from $(w,y)$. So, once Bob, who has $y^n$ as his input, has recovered $w^n$, he can sample $z_i, i\in[n]$, according to $p_{Z|XY}(z_i|x_i,y_i)$ for any $x_i\in w_i$ with $p(x_i,y_i)>0$. Privacy against Alice follows from the fact that the message that Alice sends to Bob is a deterministic function of her input.
\end{proof}
\section*{Acknowledgements}
The author would like to gratefully thank Vinod M. Prabhakaran for many helpful discussions and his help in improving the presentation of this paper.
The research was supported in part by a Microsoft Research India Ph.D. Fellowship and ERC under the EU's Seventh Framework Programme (FP/2007-2013) ERC Grant Agreement n. 307952.
\bibliographystyle{IEEEtran}
\bibliography{crypto}

\appendices
\section{Proofs Omitted from \Sectionref{twoSidedPrivacy}}\label{app:app-twoSidedPrivacy}

\begin{lem}\label{lem:compactness}
For every $\epsilon>0$, the set $\mathcal{S}_{\epsilon}$ defined in \eqref{eq:set-for-ss} is compact, i.e., closed and bounded.
\end{lem}
\begin{proof}
{\bf Boundedness:} Since $\X$ and $\Y$ are of finite cardinality, it follows from the Fenchel-Eggleston's strengthening of Carath\'eodory's theorem \cite[pg. 310]{CsiszarKorner81} that it is sufficient to consider the $p_{M\hat{Z}|XY}$'s in $\mathcal{S}_{\epsilon}$, where $M$ takes at most $\X\cdot\Y+2$ values. Now, since each $p_{M\hat{Z}|XY}$ is a finite dimensional probability vector, it follows that the set $\mathcal{S}_{\epsilon}$ is bounded.

{\bf Closedness:} Consider any convergent sequence $(p_{M\hat{Z}|XY}^{(k)})_{k\in\mathbb{N}}$, where $p_{M\hat{Z}|XY}^{(k)}\in\mathcal{S}_{\epsilon}, \forall k\in\mathbb{N}$. Let $q_{M\hat{Z}|XY}=\lim_{k\to\infty}p_{M\hat{Z}|XY}^{(k)}$. In the following we abbreviate $p_{M\hat{Z}|XY}^{(k)}$ and $q_{M\hat{Z}|XY}$ by $p^{(k)}$ and $q$, respectively. To prove closedness of $\mathcal{S}_{\epsilon}$ we need to show that $q\in\mathcal{S}_{\epsilon}$, i.e., $I(M;Y,\hat{Z}|X)|_q\leq \epsilon$, $I(M;X|Y,\hat{Z})|_q\leq \epsilon$, and $\mathbb{E}_{XY}||q_{\hat{Z}|XY}-q_{Z|XY}||_1 \leq \epsilon$ hold. Consider the first term $I(M;Y,\hat{Z}|X)$. Since mutual information is a continuous function of the distribution, it follows from the definition of continuous function, that if $(p^{(k)})_{k\in\mathbb{N}}$ converges to $q$, then $(I(M;Y,\hat{Z}|X)_{p^{(k)}})_{k\in\mathbb{N}}$ must converge to $I(M;Y,\hat{Z}|X)_q$. And by the definition of convergence, if  $I(M;Y,\hat{Z}|X)_{p^{(k)}}\leq \epsilon, \forall k\in\mathbb{N}$, then it cannot be the case that $I(M;Y,\hat{Z}|X)_q>\epsilon$. Similarly we can prove that $I(M;X|Y,\hat{Z})|_q\leq \epsilon$. Since $L_1$-norm also a continuous function, $\mathbb{E}_{XY}||q_{\hat{Z}|XY}-q_{Z|XY}||_1 \leq \epsilon$ also follows similarly. So we have $q\in\mathcal{S}_{\epsilon}$, and this concludes the proof.
\end{proof}

\begin{lem}\label{lem:as-implies-ss}
Let $\epsilon>0$ be a fixed constant. If an asymptotically secure protocol $\Pi_n$ satisfies \eqref{eq:as_2_correctness}-\eqref{eq:as_2_privacy-bob}, then there exists another protocol $\Pi_{\epsilon'}$, where $\epsilon'\to0$ as $\epsilon\to0$, that satisfies \eqref{eq:ss_2_correctness}-\eqref{eq:ss_2_privacy-bob}.
\end{lem}
\begin{proof}
{\allowdisplaybreaks
First we show that \eqref{eq:as_2_privacy-alice} implies \eqref{eq:ss_2_privacy-alice}, \eqref{eq:as_2_privacy-bob} implies \eqref{eq:ss_2_privacy-bob}, and \eqref{eq:as_2_correctness} implies \eqref{eq:ss_2_correctness}.
\begin{align*}
n\epsilon &\geq I(M;Y^n,\hat{Z}^n|X^n) \\
&= \sum_{i=1}^n I(M;Y^n,\hat{Z}_i|\hat{Z}^{i-1},X^n) \\
&\geq \sum_{i=1}^n I(M;Y_i,\hat{Z}_i|Y_{-i},\hat{Z}^{i-1},X^n) \\
&= \sum_{i=1}^n I(M,X_{-i}Y_{-i}\hat{Z}^{i-1};Y_i,\hat{Z}_i|X_i) \\
&\hspace{3cm} - I(X_{-i}Y_{-i}\hat{Z}^{i-1};Y_i,\hat{Z}_i|X_i) \\
&\geq \sum_{i=1}^n I(\underbrace{M,Y_{-i}\hat{Z}^{i-1}}_{=\ M_i};Y_i,\hat{Z}_i|X_i) \\
&\hspace{3cm} - \underbrace{\sum_{i=1}^n I(X_{-i}Y_{-i}\hat{Z}^{i-1};Y_i,\hat{Z}_i|X_i)}_{\leq\ n\epsilon_1, \text{ By \Lemmaref{small}}} \\
&\geq \sum_{i=1}^n I(M_i;Y_i,\hat{Z}_i|X_i) - n\epsilon_1 \\
&= n\cdot \sum_{i=1}^n \frac{1}{n}\cdot I(M_i;Y_i,\hat{Z}_i|X_i,T=i) - n\epsilon_1 \\
&= n\cdot I(M_T;Y_T,\hat{Z}_T|X_T,T) - n\epsilon_1 \\
&= n\cdot I(\underbrace{M_T,T}_{=\ M'};Y_T,\hat{Z}_T|X_T) - n\epsilon_1
\end{align*}
where $T$ is distributed uniformly in $\{1,2,\hdots,n\}$ and is independent of $(M,X^n,Y^n,\hat{Z}^n)$. So we have $I(M';Y_T,\hat{Z}_T|X_T) \leq \epsilon+\epsilon_1$.
\begin{align*}
n\epsilon &\geq I(M;X^n|Y^n,\hat{Z}^n) \\
&= \sum_{i=1}^n I(M;X_i|X^{i-1},Y^n,\hat{Z}^n) \\
&= \sum_{i=1}^n I(M,X^{i-1},Y_{-i},\hat{Z}_{-i};X_i|Y_i,\hat{Z}_i) \\
&\hspace{3cm}- \underbrace{\sum_{i=1}^n I(X^{i-1},Y_{-i},\hat{Z}_{-i};X_i|Y_i,\hat{Z}_i)}_{\leq\ n\epsilon_2, \text{ By \Lemmaref{small}}} \\
&\geq \sum_{i=1}^n I(\underbrace{M,Y_{-i},\hat{Z}^{i-1}}_{=\ M_i};X_i|Y_i,\hat{Z}_i) - n\epsilon_2 \\
&= n\cdot\sum_{i=1}^n \frac{1}{n}\cdot I(M_i;X_i|Y_i,\hat{Z}_i,T=i) - n\epsilon_2 \\
&= n\cdot I(M_T;X_T|Y_T,\hat{Z}_T,T) - n\epsilon_2 \\
&= n\cdot I(\underbrace{M_T,T}_{=\ M'};X_T|Y_T,\hat{Z}_T) - n\epsilon_2
\end{align*}
So we have $I(M';X_T|Y_T,\hat{Z}_T)\leq \epsilon+\epsilon_2$.
\begin{align*}
&\mathbb{E}_{X_TY_T}[||p_{\hat{Z}_T|X_TY_T}-p_{Z_T|X_TY_T}||_1] \\
&= \sum_{x,y}p_{X_TY_T}(x,y)||p_{\hat{Z}_T|X_TY_T}-p_{Z_T|X_TY_T}||_1 \\
&= \sum_{x,y}p_{X_TY_T}(x,y)\sum_{z}|p_{\hat{Z}_T|X_TY_T}(z|x,y)-p_{Z_T|X_TY_T}(z|x,y)| \\
&= \sum_{x,y,z}|p_{X_TY_T\hat{Z}_T}(x,y,z)-p_{X_TY_TZ_T}(x,y,z)| \\
&= \sum_{x,y,z}|\displaystyle \sum_{\substack{x^n,y^n,z^n:\\x_T=x,\\y_T=y,\\z_T=z}}p_{X^nY^n\hat{Z}^n}(x^n,y^n,z^n)-p_{X^nY^nZ^n}(x^n,y^n,z^n)| \\
&\leq \sum_{x^n,y^n,z^n}|p_{X^nY^n\hat{Z}^n}(x^n,y^n,z^n)-p_{X^nY^nZ^n}(x^n,y^n,z^n)| \\
&= \sum_{x^n,y^n}p_{X^nY^n}(x^n,y^n)\times\\
&\hspace{1.5cm} \sum_{z^n}|p_{\hat{Z}^n|X^nY^n}(z^n|x^n,y^n)-p_{Z^n|X^nY^n}(z^n|x^n,y^n)| \\
&= \sum_{x^n,y^n}p_{X^nY^n}(x^n,y^n)||p_{\hat{Z}^n|X^nY^n}-p_{Z^n|X^nY^n}||_1 \\
&= \mathbb{E}_{X^nY^n}[||p_{\hat{Z}^n|X^nY^n}-p_{Z^n|X^nY^n}||_1] \\
&\leq \epsilon.
\end{align*}
Now take $\epsilon'=\max\{\epsilon+\epsilon_1, \epsilon+\epsilon_2\}$. Clearly, $\epsilon'\to0$ as $\epsilon\to0$.

To get a protocol $\Pi_{\epsilon'}$ for a given $(p_{XY},p_{Z|XY})$ that satisfies \eqref{eq:ss_2_correctness}-\eqref{eq:ss_2_privacy-bob}, we take a random index $T\in\{1,2,\hdots,n\}$ and fix $(X_T,Y_T)=(X,Y)$. Next we pick $(n-1)$ i.i.d. copies of $(X,Y)$ and fill the remaining indices of $(X^n,Y^n)$ by these $(n-1)$ copies. Now we have $(X^n,Y^n)$ with $(X_T,Y_T)=(X,Y)$. Run the protocol $\Pi_n$ (which is for asymptotically secure computation) with these inputs. We get the output $\hat{Z}^n$. Set $Z=\hat{Z}_T$. By the above calculations, it follows that $\Pi_{\epsilon'}$ is a protocol for $(p_{XY},p_{Z|XY})$ that satisfies \eqref{eq:ss_2_correctness}-\eqref{eq:ss_2_privacy-bob}. This concludes the proof.
}
\end{proof}

\begin{lem}\label{lem:small}
\begin{align}
\sum_{i=1}^n I(X_{-i},Y_{-i},\hat{Z}^{i-1};Y_i,\hat{Z}_i|X_i) &\leq n\epsilon_1, \label{eq:small-1}\\
\sum_{i=1}^n I(X^{i-1},Y_{-i},\hat{Z}_{-i};X_i|Y_i,\hat{Z}_i) &\leq n\epsilon_2, \label{eq:small-2}\\
\sum_{i=1}^n I(Y_{-i},\hat{Z}^{i-1};X_i,\hat{Z}_i|Y_i) &\leq n\epsilon_3, \label{eq:small-3}
\end{align}
where $\epsilon_1,\epsilon_2,\epsilon_3\to 0$ as $\epsilon\to0$.
\end{lem}
\begin{proof}
Consider the $i$th term $I(X_{-i},Y_{-i},\hat{Z}^{i-1};Y_i,\hat{Z}_i|X_i)$ in the left hand 
side of \eqref{eq:small-1}. Notice that if we replace $\hat{Z}$ by $Z$, then this term will 
vanish. Now, since mutual information is a continuous function of the distribution, and 
$||p_{X^nY^n\hat{Z}^n}-p_{X^nY^nZ^n}||_1\leq \epsilon$, we have 
$I(X_{-i},Y_{-i},\hat{Z}^{i-1};Y_i,\hat{Z}_i|X_i)\leq \epsilon_{1i}$, 
where $\epsilon_{1i}\to 0$ as $\epsilon\to 0$. This is true for every $i\in[n]$.
Now let $\epsilon_1:=\max_{i\in [n]}\epsilon_{1i}$. Clearly $\epsilon_1\to0$ as $\epsilon\to0$.
Similar we can prove \eqref{eq:small-2} and \eqref{eq:small-3}.
\end{proof}

\begin{proof}[Remaining proof of \Theoremref{as_2_rate-region}]

{\bf Explanation for \eqref{eq:as2-interim1}:} For a fixed $(p_{XY},p_{Z|XY})$, let us define the following function:
\begin{align}
R_{\epsilon}(p_{XY},p_{Z|XY}) := \min_{\substack{p_{U\hat{Z}|XY}: \\ I(U;Y,\hat{Z}|X)\leq\epsilon \\ I(U;X|Y,\hat{Z})\leq\epsilon \\ I(\hat{Z};X|U,Y)\leq\epsilon \\ \mathbb{E}||p_{\hat{Z}|XY}-p_{Z|XY}||_1\leq\epsilon}} I(U;X,\hat{Z}|Y). \label{eq:function-ss}
\end{align}
Note that the expression in \eqref{eq:as2-interim} can be lower-bounded by $R_{\epsilon}(p_{XY},p_{Z|XY})-\epsilon_3$, because we relax one of the Markov chains in the definition of $R_{\epsilon}$. Also note that $R_0$ corresponds to perfectly secure protocols. To prove the inequality in \eqref{eq:as2-interim1}, it suffices to show that the function $R_{\epsilon}(p_{XY},p_{Z|XY})$ is right continuous at $\epsilon=0$. In order to show this we define a rate-region tradeoff corresponding to this function as follows:
\begin{align}
&\mathcal{T}(p_{XY},p_{Z|XY}) := \{(\epsilon,r): \exists p_{U\hat{Z}|XY} \text{ s.t. } I(U;X,\hat{Z}|Y)\leq r, \notag\\ &\hspace{2.5cm} I(U;Y,\hat{Z}|X)\leq\epsilon, I(U;X|Y,\hat{Z})\leq\epsilon, \notag\\&\hspace{2cm} I(\hat{Z};X|U,Y)\leq\epsilon, \mathbb{E}||p_{\hat{Z}|XY}-p_{Z|XY}||_1\leq\epsilon \}. \label{eq:rate-region-tradeoff}
\end{align}
For sake of simplicity, in the following we denote $R_{\epsilon}(p_{XY},p_{Z|XY})$ and $\mathcal{T}(p_{XY},p_{Z|XY})$ by $R_{\epsilon}$ and $\mathcal{T}$, respectively.
First we show that the above defined set $\mathcal{T}$ is a closed set. Then using the closedness of $\mathcal{T}$, we prove that the function $R_{\epsilon}$ is right continuous at $\epsilon=0$, which suffices to explain \eqref{eq:as2-interim1}.
\begin{lem}\label{lem:rate-region-closedness}
$\mathcal{T}(p_{XY},p_{Z|XY})$, as defined in \eqref{eq:rate-region-tradeoff}, is a closed set.
\end{lem}
\begin{proof}
Let $\mathcal{P}_{XY}$ denote the set of all conditional distributions $p_{U\hat{Z}|XY}$. Since $\X$ and $\Y$ are finite alphabets, it follows from the Fenchel-Eggleston's strengthening of Carath\'eodory's theorem \cite[pg. 310]{CsiszarKorner81} that we can upper-bound the cardinality of $|\mathcal{U}|$ by $|\X|\cdot|\Y|+2$. This implies that the set $\mathcal{P}_{XY}$ is compact, i.e., closed and bounded. 
For a fixed $(p_{XY},p_{Z|XY})$, consider the following function:
\begin{align*}
h(p_{U\hat{Z}|XY})&=(I(U;X,\hat{Z}|Y),I(U;Y,\hat{Z}|X),I(U;X|Y,\hat{Z}),\\
&\hspace{1.5cm} I(\hat{Z};X|U,Y),\mathbb{E}||p_{\hat{Z}|XY}-p_{Z|XY}||_1).
\end{align*}
Note that $\mathcal{T}(p_{XY},p_{Z|XY})$ is the increasing hull of \image\ of the above defined function $h$, where increasing hull of a set $\mathcal{S}\subseteq \mathbb{R}^5$ is defined as $\{(a_1,a_2,a_3,a_4,a_5)\in\mathbb{R}^5:\exists (a_1',a_2',a_3',a_4',a_5')\in\mathcal{S} \text{ s.t. } a_i'\leq a_i\ \forall i\in[5]\}$. Since increasing hull of a closed set is also a closed set, it is enough to prove that \image(h) is a closed set. Since mutual information and $L_1$-norm are continuous functions of distribution, and image of a compact set under a continuous function is always a compact set, and therefore closed. Now, closedness of \image(h) follows from the fact that the function $h$ is a continuous function of $p_{U\hat{Z}|XY}\in \mathcal{P}_{XY}$.
\end{proof}
\begin{lem}\label{lem:right-continuous}
If $(p_{XY},p_{Z|XY})$ is computable with asymptotic security, then the function $R_{\epsilon}$ defined in \eqref{eq:function-ss} is right continuous at $\epsilon=0$.
\end{lem}
\begin{proof}
We prove this using \Lemmaref{rate-region-closedness} and the fact that $R_0$ is bounded, which follows from the fact (\Theoremref{as_characterization}) that if $(p_{XY},p_{Z|XY})$ is computable with asymptotic security then it can be computable with perfect security.

Suppose, to the contrary, that $R_{\epsilon}$ is not right continuous at $\epsilon=0$. This implies that there exists a monotone decreasing sequence $(\epsilon_1,\epsilon_2,\hdots)$ satisfying $\lim_{i\to\infty}\epsilon_i=0$ and a constant $\beta>0$ such that $R_{\epsilon_i}\leq R_0-\beta, \forall i\in\mathbb{N}$. Observe that $R_{\epsilon_i}$ is a non-decreasing sequence that is bounded above by $R_0$ (and $R_0$ is finite by \Theoremref{as_characterization}). Since every monotone increasing sequence bounded above is convergent, it follows that $R_{\epsilon_i}$ is convergent. Let $l=\lim_{i\to\infty}R_{\epsilon_i}$. We have $l\leq R_0-\beta<R_0$. Since $\mathcal{T}$ is a closed set (i.e., it contains all its limit points), we have $l\in\mathcal{T}$. But this contradicts the fact that $R_0$ is the minimum value $r$ such that $(r,0,0,0,0)\in\mathcal{T}$.
\end{proof}
\end{proof}

\section{Details Omitted from \Sectionref{oneSidedPrivacy}}\label{app:app-oneSidedPrivacy}
\begin{lem}\label{lem:cond-graph-entropy-equi}
\[\displaystyle \min_{\substack{p_{U|X}: \\ U-X-(Y,Z) \\ Z-(U,Y)-X}} I(U;X|Y) = \displaystyle \min_{\substack{p_{W|X}: \\ W-X-Y \\ X\in W\in \varGamma(G)}} I(W;X|Y).\]
\end{lem}
\begin{proof}
$\leq:$ Suppose $p_{W|X}$ achieve the minimum on the right hand side such that $W\in\varGamma(G)$ and $X\in W$. 
Consider any independent set, say $w\in G$. By definition of $G$, for all $(y,z)\in\mathcal{Y}\times\mathcal{Z}$, $p_{Z|X,Y}(z|x,y)$ is the same for all $x\in w$ with $p_{X,Y}(x,y)>0$. 
Hence $p_{Z|XY}(z|x,y)$ can be uniquely determined from $(w,y,z)$ where $x\in w$, and so $Z-(W,Y)-X$ holds. 
Since $p_{W|X}=p_{W|XYZ}$ holds by definition, we have $W-X-(Y,Z)$.
Now the inequality follows by taking $U=W$.

$\geq:$ This is slightly more involved than the similar inequality (corresponding to deterministic functions) proved by Orlitsky and Roche \cite{OrlitskyRoche}, because we are dealing with randomized functions. Suppose $p_{U|X}$ achieve the minimum on the left hand side such that $U-X-(Y,Z)$ and $Z-(U,Y)-X$ hold. 
Now define a random variable $W$ as a function of $U$ in the following way: $w=w(u):=\{x:p_{U,X}(u,x)>0\}$. 
We need to show that the induced conditional distribution $p_{W|X}$ satisfies the following three conditions: 1) the Markov chain $W-X-Y$ holds; 2) $X\in W$, i.e., $p_{X,W}(x,w)>0 \implies x\in w$; and 3) $W$ is an independent set in $G=(X,E)$.

1) The Markov chain $W-X-Y$ holds because $U-X-(Y,Z)$ holds and $W$ is a function of $U$.
2) If $p_{W,X}(w,x)>0$ then $\exists u$ s.t. $w=w(u)$ and $p_{U,X}(u,x)>0$, which implies that $x\in w$.
3) To prove that $w$ is an independent set we suppose, to the contrary, that $w$ is not an independent set, which means that $\exists x,x',u$, where $w=w(u), x,x'\in w$, $\{x,x'\}\in E$, and such that $p_{U|X}(u|x)\cdot p_{U|X}(u|x')>0$. 
By definition of $G$, $\{x,x'\}$ being an edge in $E$ implies that $\exists (y,z)\in\mathcal{Y}\times\mathcal{Z}$ such that $p_{Z|XY}(z|x,y)\neq p_{Z|X,Y}(z|x',y)$ and $p_{XY}(x,y)\cdot p_{XY}(x',y)>0$. 
Consider the above $(u,x,y,z)$. We can expand $p_{UZ|XY}(u,z|x,y)$ in two different ways. The first expansion is as follows:
\begin{align}
p_{UZ|XY}(u,z|x,y) &= p_{Z|XY}(z|x,y)p_{U|XYZ}(u|x,y,z) \nonumber \\
&= p_{Z|XY}(z|x,y)p_{U|X}(u|x), \label{eq:cutset_equality1}
\end{align}
where in the second equality we use $U-X-(Y,Z)$. We can expand $p_{UZ|XY}(u,z|x,y)$ in the following way also:
\begin{align}
p_{UZ|XY}(u,z|x,y) &= p_{U|XY}(u|x,y)p_{Z|U,XY}(z|u,x,y) \nonumber \\
&= p_{U|X}(u|x)p_{Z|UY}(z|u,y), \label{eq:cutset_equality2}
\end{align}
where in the second equality we use $U-X-Y$ and $Z-(U,Y)-X$. 
Now comparing \eqref{eq:cutset_equality1} and \eqref{eq:cutset_equality2} and since $p_{U|X}(u|x)>0$, we get $p_{Z|XY}(z|x,y) = p_{Z|UY}(z|u,y)$.
Applying the same reasoning with $(u,x',y,z)$ we get $p_{Z|XY}(z|x',y) = p_{Z|UY}(z|u,y)$, which leads to the contradiction that $p_{Z|XY}(z|x,y)\neq p_{Z|XY}(z|x',y)$.

\end{proof}
\end{document}